\documentclass[a4paper]{article}[11pts]

\usepackage[english]{babel}
\usepackage[utf8x]{inputenc}
\usepackage[T1]{fontenc}

\normalsize

\usepackage[a4paper,top=1in,bottom=1in,left=1in,right=1in,marginparwidth=1in]{geometry}

\usepackage{setspace}
\usepackage{amsmath}
\usepackage{amssymb} 
\usepackage{amsthm}
\usepackage{amsfonts, mathtools}
\usepackage{algorithm,algpseudocode}
\usepackage{bm}
\usepackage{comment}
\usepackage{graphicx}
\usepackage{multirow, booktabs}
\usepackage{caption}
\usepackage{subcaption}
\usepackage[colorinlistoftodos]{todonotes}
\usepackage[colorlinks=true, allcolors=blue]{hyperref}

\newcommand{\E}{\mathbb{E}}

\newcommand{\prob}{\mathbb{P}}

\DeclarePairedDelimiter\ceil{\lceil}{\rceil}

\newtheorem{proposition}{Proposition}
\newtheorem{lemma}{Lemma}

\newtheorem{theorem}{Theorem}
\newtheorem{corollary}{Corollary}
\newtheorem{definition}{Definition}

\usepackage{natbib}

\title{Online Bin Packing with Known $T$}
\author{}
\date{}

\author{Shang Liu \and Xiaocheng Li 
\footnote{The authors thank Varun Gupta for bringing the problem to our attention.}}
\date{\small 
Imperial College Business School, Imperial College London\\
$\{$s.liu21, xiaocheng.li$\}$@imperial.ac.uk\\
}

\begin{document}
\maketitle

\onehalfspacing

\begin{abstract}
In the online bin packing problem, a sequence of items is revealed one at a time, and each item must be packed into an available bin instantly upon its arrival. In this paper, we revisit the problem under a setting where the total number of items $T$ is known in advance, also known as the closed online bin packing problem. Specifically, we study both the stochastic model where the item sizes are drawn independently from an unknown distribution and the random permutation model where the item sizes may be adversarially chosen, but the items arrive in a randomly permuted order. We develop and analyze an adaptive algorithm that solves an offline bin packing problem at geometric time intervals and uses the offline optimal solution to guide the online packing decisions. Under both models, we show that the algorithm achieves $C \sqrt{T}$ regret (in terms of the number of used bins) compared to the hindsight optimal solution, where $C$ is a universal constant ($\le 13$) that bears no dependence on the underlying distribution or the item sizes. The result shows the lower bound barrier of $\Omega(\sqrt{T \log T})$ in \citep{shor1986average} can be broken with solely the knowledge of the horizon $T$, but without a need of knowing the underlying distribution. As to the algorithm analysis, we develop a new approach to analyzing the packing dynamic using the notion of \textit{exchangeable 
random variables}. The approach creates a symmetrization between the offline solution and the online solution, and it is used to analyze both the algorithm performance and the various benchmarks related to the bin packing problem. For the latter one, our analysis provides an alternative (probably simpler) treatment and tightens the analysis of the asymptotic benchmark of the stochastic bin packing problem in \citep{rhee1989optimal2, rhee1989optimal3}. As the analysis only relies on a symmetry between the offline and online problems, the algorithm and benchmark analyses can be naturally extended from the stochastic model to the random permutation model.
\end{abstract}

\section{Introduction}

The bin packing problem is a fundamental optimization problem that concerns the optimal packing of a collection of items into a finite number of bins such that the number of bins used is minimized. It arises in many applications such as remnant scheduling/cutting stock \citep{adelman1999price}, resource allocation \citep{cohen2019overcommitment}, transportation logistics \citep{gupta2020interior},  etc. The \textit{online bin packing} problem considers an online setting where the items arrive sequentially, and each item must be packed instantly upon its arrival. The problem is commonly studied under two settings: average case and worst case. The average case always assumes the item sizes are sampled independently and identically from some distribution, while the worst case allows the item sizes and their arriving orders to be adversarially chosen. In this paper, we focus on the average case and consider two specific models for item size sequence: the stochastic model and the random permutation model. Under the stochastic model, the item sizes are sampled independently from an unknown distribution, and we do not impose any further assumption on the distribution. Under the random permutation model, the item sizes can be chosen adversarially, but the items arrive in a randomly permuted order. The performance of an online algorithm is measured by \textit{regret}, formally defined by the gap between the number of bins used by the online algorithm, and the optimal number of bins used by an offline algorithm that knows all the item sizes in advance. 

\begin{table}[ht!]
\centering
\small
\renewcommand{\arraystretch}{1.5}
\begin{tabular}{c|c|c|c}
\toprule
     & Regret  & Distr. & Algorithm \& Remarks  \\      \hline
    \cite{shor1986average} & $\Omega(\sqrt{T \log T})$ & Unif$[0,1]$ & Lower bound \\      \hline
        \cite{shor1986average, asgeirsson2002closed} & $\Theta(\sqrt{T})$ & Unif$[0,1]$ &   Best Fit; Known $T$    \\     \hline
    \cite{shor1991pack} & $O(\sqrt{T\log T})$ & Unif$[0,1]$ &  Best Fit   \\ 
                \hline
\cite{rhee1993line2,rhee1993line1}  & $K\sqrt{T}\log^{3/4} T$ & General & Double-overflow; unspecified constant $K$ \\      \hline
 \cite{csirik2006sum} & $B\sqrt{T}$ & Int. supp. & Sum-of-squares; bin size $B$ \\      \hline
 \cite{gupta2020interior} & $B\sqrt{T}$ & Int. supp. &  Lagrangian-based; bin size $B$ \\\hline
 \cite{banerjee2020uniform} & M & Int. supp. & Re-solving; Known $T$; problem-dependent $M$ \\ \hline
Ours & $C\sqrt{T}$ & General  & Adaptive; Known $T$; $C\le 11$\\\hline
Ours & $C\sqrt{T}$ & Ran. Perm. & Adaptive; Known $T$; $C\le 13$\\
 \bottomrule
    \end{tabular}
    \caption{Summary of existing results on average-case online bin packing}
    \label{tab:result}
\end{table}

Table \ref{tab:result} summarizes some existing results on the average-case online bin packing problem. \cite{shor1986average} first establishes a lower bound of $\Omega(\sqrt{T\log T})$ for the problem. The lower bound by \cite{shor1986average} considers an item size distribution of Uniform$[0,1]$, and the derivation relies on the critical assumption of not knowing $T$. In the following, we compare our result against the existing upper bounds. 

\textbf{Uniform$[0,1]$ distribution.} In the same paper, \cite{shor1986average} shows that the Best Fit algorithm achieves $\Theta(\sqrt{T})$ regret bound if the underlying item size distribution is Uniform$[0,1]$ and the horizon length $T$ is known in advance. In a similar vein, \cite{asgeirsson2002closed} also considers the problem with a known $T$ and analyzes a variant called Closed Best Fit that also achieves $\Theta(\sqrt{T})$ regret. \cite{shor1991pack} develops an algorithm that achieves $O(\sqrt{T\log T})$ regret but does not need the knowledge of $T$. The distribution of Uniform$[0,1]$ plays a critical role in all these analyses. Specifically, when the item size distribution is $\mathrm{Unif}[0,1]$, the (near) optimal solution of the bin packing problem can be viewed as a matching of items smaller than $1/2$ with items larger than $1/2$, for both the case of known $T$ and unknown $T$. Thus the analysis is usually based on an argument that each bin ideally should contain exactly two items: one of size $1/2 - \epsilon$ and the other of size $1/2 + \epsilon$ for some $\epsilon\in(0,1/2)$. When the item size follows a general distribution, the optimal offline solution may have a more complicated structure which prohibits such analysis. 

\textbf{Distribution with integer support.} Another stream of works \citep{csirik2006sum, gupta2020interior, banerjee2020uniform} assume the bin size to be $B\in \mathbb{N}$ and the item size distribution is integer-valued and supported on $\{1,...,B-1\}$. A backbone for the analyses in \citep{csirik2006sum, gupta2020interior} is the categorization of the item distribution into Linear Waste, Perfectly Packable, and Perfectly Packable with Bounded Waste by \cite{courcoubetis1986necessary}. The categorization is based on the optimal objective value of the \textit{Waste Linear Program} (LP) associated with item size distribution which measures the expected waste bin space in the offline optimal solution. Both \cite{csirik2006sum} and \cite{gupta2020interior} use some potential function as a heuristic to guide the online packing decision and both achieve $O(B\sqrt{T})$ regret for the problem. The algorithms in these two papers do not rely on the knowledge of underlying distribution or the horizon length $T$, and are efficient to implement, but we believe the factor $B$ in the regret bound is hard to get rid of for these algorithms. \cite{banerjee2020uniform} take a different approach by viewing the online bin packing problem as a dynamic programming problem. The algorithm by \cite{banerjee2020uniform} requires the knowledge of the item distribution and the horizon length $T$, solves an integer LP at each time period, and uses its optimal solution to guide the online decision. While the other regret bounds in Table \ref{tab:result} are all problem-independent, \cite{banerjee2020uniform} provide a problem-dependent regret bound that bears no dependency on $T$ but involves a number of parameters related to the item size distribution and bin size $B$. As far as we understand, there is no clear way to transform the problem-dependent regret therein into a problem-independent one. 

\textbf{General distribution.} To the best of our knowledge, the only existing results that consider a general item size distribution are done by Rhee and Talagrand in a series of works \citep{rhee1988exact, rhee1988optimal1, rhee1989optimal2, rhee1989optimal3, rhee1993line2, rhee1993line1}. The general distribution may be continuous like Uniform$[0,1]$, discrete like the integer-valued case or a mixture of them. These works consist of two parts: the first part \citep{rhee1988exact, rhee1988optimal1, rhee1989optimal2, rhee1989optimal3} studies the stochastic bin packing problem, and the second part \citep{rhee1993line2, rhee1993line1} builds upon the result and concerns the online bin packing problem. For the first part, the stochastic bin packing packing can be viewed as the analysis of the offline benchmark in the online bin packing problem. Specifically, suppose $\mathcal{X}_T=\{X_1,...,X_T\}$ is a set of item sizes sampled independently from some distribution $F$. The stochastic bin packing problem concerns the analysis of the $E[\text{OPT}(\mathcal{X}_T)]$. \cite{rhee1988optimal1} characterizes the class of distributions such that $\lim_{T\rightarrow \infty}\mathrm{OPT}(\mathcal{X}_T)/T = \E[X_t]$. This property correspond to the case of Bounded Waste or Perfectly Packable in the \cite{csirik2006sum}. \cite{rhee1989optimal2} further provide a bound that $|\E[\mathrm{OPT}(\mathcal{X}_T)] - T \cdot c| \leq K \sqrt{T \log T}$, where $c$ is the asymptotic benchmark to be defined later in our paper. In a subsequent work \citep{rhee1989optimal3}, the authors further improve the bound but in an asymptotic sense. Specifically, they show that the factor of $\log T$ is redundant by proving that the asymptotic distribution of $(\mathrm{OPT}(\mathcal{X}_T) - T \cdot c)/\sqrt{T}$ exists and is specified by a Gaussian process related to the item size distribution. Compared to these results on the stochastic bin packing problem, we explicitly show that $|\E[\mathrm{OPT}(\mathcal{X}_T)] - T \cdot c| \leq C \sqrt{T}$ for some $C\le 11$. Our result is finite-sample and bears no dependency on the underlying distribution of $X_t$. For the second part, \cite{rhee1993line1, rhee1993line2} build upon the previous work and provide an online bin packing algorithm under the stochastic model where $T$ is unknown. The regret bound for their algorithm is $O(\sqrt{T} \log^{\frac34}T)$. As noted in an earlier paper \citep{rhee1988exact}, the regret analysis relies on the up-right matching problem thus the bound cannot be further improved. Our algorithm can be viewed as a simplification of the algorithm therein by utilizing the knowledge of $T$, though the analyses are completely different. 


\subsection{Our Contribution}

To summarize, the contribution of our paper lie in the following three aspects:

\begin{itemize}
    \item Online bin packing with known $T$. We develop a natural adaptive algorithm that achieves $O(\sqrt{T})$ regret bound for the online bin packing problem with known horizon length $T$. The algorithm and its regret bound apply to any arbitrary item size distribution. 
    \item Benchmark analysis for stochastic bin packing. The key to our regret analysis is the usage of the maximal inequality of the exchangeable random variables. A byproduct of our analysis simplifies and tightens the existing bound on the stochastic bin packing problem. 
    \item Online bin packing under random permutation model. Our algorithm and analysis extend to the random permutation model as well. As far as we know, our paper first shows that sublinear regret is achievable for the online bin packing problem under the random permutation model. 
\end{itemize}

\subsection{Other Literature}

Our algorithm naturally solves a number of offline problems and uses the offline solution to guide the online decisions. In this sense, the result complements to the recent works on using offline data for online decision making \citep{bu2020online, wang2020measuring}. On a high level, these works focus on how the offline data can be used to refine the knowledge of the underlying distribution, while our analysis detaches from the learning of the underlying distribution and only takes advantage of a symmetry between the offline and online data. 

Another stream of literature on online bin packing takes an algorithm-centric perspective and analyzes the performance of certain heuristic algorithms. For example, Best Fit \citep{johnson1973near}, First Fit \citep{johnson1973near, garey1976resource}, Next Fit \citep{johnson1973near}, Any Fit  \citep{johnson1974fast}, and also the survey paper \citep{coffman1996approximation}. Our algorithm, though only requires solving $O(\log T)$ many offline oracles and permits LP-based approximate solution, is nonetheless efficient compared to these heuristics. In this light, our analysis has a more information-theoretical focus and aims to provide new insights and develop analytical tools for the problem. 

The rest of the paper is organized as follows. In Section \ref{sec_model}, we present our main algorithm and the first step of our regret analysis. In Section \ref{sec_stoc}, we analyze the stochastic bin packing benchmarks and derive a regret bound for our algorithm under the stochastic model. In Section \ref{sec_RP_AA}, we analyze our algorithm under the random permutation model and show its compatibility with an approximate solution to the offline problem. In Section \ref{sec_experi}, we conclude our discussion with an LP-based adaptive algorithm that gives better numerical performance and leaves its analysis as a future open question.

\section{Model Setup and Algorithm}

\label{sec_model}

We are given an infinite collection of empty bins, into which a sequence of $T$ items of sizes $X_1,...,X_T$ is to be packed with the goal of minimizing the number of bins used. The items arrive in an online fashion and each item needs to be packed instantly after its arrival and without the knowledge of future item sizes. It does not hurt to assume the empty bins have capacity $1$ and item size $X_t\in(0,1)$ for all $t\in[T].$ We consider both the \textit{stochastic} model where the item sizes $X_t$'s are identically and independently drawn from an unknown distribution $F$ and the \textit{random permutation} model where the item sizes $X_t$'s can be chosen adversarially but they arrive in a randomly permuted order. In this section, the algorithm and the analysis work for both models; in the following sections, we will separate the discussion for the two models. 

We measure the performance of an online algorithm by the gap between the number of bins used by the algorithm and the hindsight/offline optimal --  the number of bins used by an optimal packing scheme which knows all the item sizes. Specifically, we define the regret of an algorithm for a problem instance $\mathcal{X}_T=\{X_1,...,X_T\}$ as 
$$\text{Reg}(\mathcal{X}_T) \coloneqq N^\pi(\mathcal{X}_T) - \text{OPT}(\mathcal{X}_T)$$
where $\text{OPT}(\mathcal{X}_T)$ denotes the number of bins used by an offline optimal solution and $N^\pi(\mathcal{X}_T)$ denotes the number of bins used by an online algorithm $\pi.$ Throughout the paper, we use $\mathcal{X}_t$ to denote the set $\{X_1,...,X_t\}$ and $\mathcal{X}_{s:t}$ to denote the set $\{X_s,...,X_t\}$ for $s\le t$. Eventually, we will aim to derive a problem-independent/worst-case regret by taking supremum with respect to the generation mechanism of $\mathcal{X}_T$. For this section, we refrain from specifying the underlying randomness so as to derive common results for both the stochastic and random permutation models.

Next, we define the offline bin packing problem by adopting its standard formulation following the presentation in \citep{williamson2011design}. Let $\mathcal{X}$ denote a finite set of items and suppose that there are $m$ different sizes of items in total, named as $0<s_1 < s_2 < \dots < s_m<1$. And there are $b_1, b_2,\dots, b_m$ items of each size accordingly, with $\sum_{i=1}^m b_i = |\mathcal{X}|$. Then consider the ways that a single bin can packed. Each type of packing can be uniquely described by a $m$-tuple $(t_1,\dots,t_m)$ where $t_i$ denotes the number of $s_i$ items. We call such an $m$-tuple as a \emph{configuration} if $\sum_{i=1}^m s_i t_i \leq 1$. Denote the total number of configurations by $N$ (possibly exponentially many of them). Let $A_1,\dots,A_N$ be a complete enumeration of them, where $a_{ij}$ denotes the $i$-th component of $A_j$.
Then the offline bin packing problem is described by the following integer LP
\begin{align*}
   \text{OPT}(\mathcal{X}) \coloneqq \min \  &\sum_{j=1}^N z_j \\
    \mathrm{s.t.}\   & \sum_{j=1}^N a_{ij}z_j \geq b_i,\quad i=1,\dots,m,\\
   &  z_j \in \mathbb{N}, \quad j=1,\dots,N, 
\end{align*}
where the decision variable $z_j$ represents the number of bins packed according to the $j$-th configuration.

The relaxed version of the above integer programming is as follows where the integer constraint is relaxed. We denote its optimal value as $\text{OPT}_f(\mathcal{X})$.
\begin{align*}
    \text{OPT}_f(\mathcal{X}) \coloneqq \min \ &\sum_{j=1}^N z_j \\
    \mathrm{s.t.}\ & \sum_{j=1}^N a_{ij}z_j \geq b_i,\quad i=1,\dots,m,\\
    & z_j \geq 0, \ \  j=1,\dots,N.
\end{align*}
The renowned result in \cite{karmarkar1982efficient} provides a relation between the bin packing problem and its relaxed LP problem. 
\begin{proposition}[\cite{karmarkar1982efficient}]
\label{prop:relaxedgap} The following inequality holds for any arbitrary set $\mathcal{X},$
$$ \mathrm{OPT}_f(\mathcal{X}) \leq \mathrm{OPT}(\mathcal{X}) \leq \mathrm{OPT}_f(\mathcal{X}) + 4 \log_2^2(|\mathcal{X}|) + 17 \log_2(|\mathcal{X}|) + 11.$$
\end{proposition}
\begin{proof}
The left-hand-side is obvious from the relaxation.
For the right-hand-side, it can be obtained from the analyses of Lemma 3 and Theorem 4 in \cite{karmarkar1982efficient}. Specifically, Algorithm 2 in \cite{karmarkar1982efficient} is capable of producing a feasible solution for the integer program within $\mathrm{OPT}_f(\mathcal{X}) + 4 \log_2^2(|\mathcal{X}|) + 17 \log_2(|\mathcal{X}|) + 11$ bins. 
\end{proof}

\begin{algorithm}[ht!]
\caption{Online Adaptive Overflow Algorithm}
\label{alg:overflow}
\begin{algorithmic}[1]
\Require Time horizon $T$. 
\Ensure A packing scheme for the $T$ items.
\State Let $K = \lceil \log_2(T) \rceil$ and $\ T_k = \lceil \frac{T}{2^{K-k}} \rceil$ for $k=0,1,...,K.$
\State Observe item $X_1$ and pack it into a new bin.
\State Initialize $t = 2$.
\For{phase $k=1,\dots, K$}
        \State Re-order the previously observed items ascendingly. Specifically,  $$\mathcal{X}_{T_{k-1}} = \{X^{(1)}, \dots , X^{(T_{k-1})}\} \text{ where } X^{(1)} \leq \dots \leq X^{(T_{k-1})}.$$
        \State Solve the offline bin packing problem  $\mathrm{OPT}(\mathcal{X}_{T_{k-1}})$ 
        \State Prepare $\mathrm{OPT}(\mathcal{X}_{T_{k-1}})$ new bins. 
        \State Record the packing location for each item in $\mathcal{X}_{T_{k-1}}$ with the mapping $$h_k: \{1,...,T_{k-1}\} \rightarrow \{1,...,\mathrm{OPT}(\mathcal{X}_{T_{k-1}})\},$$i.e., the item $X^{(s)}$ is packed in the bin $h_k(s)$.
        \State Initialize the location vacancy indicator $I_1 = \dots = I_{T_{k-1}} = 1$.
        \While{$t \leq T_k$}
            \State Observe item $X_t$.
            \State Find the packing (next vacant) location $v(X_t) = \min \{s|X^{(s)} \geq X_t, I_s = 1, s\in[T_{k-1}]\}$.
            \State \%\% The minimum of an empty set is defined to be $\infty,$ i.e., $\min \{\}=\infty.$
            \If{$v(X_t) < \infty$}
                \State Pack $X_t$ into the bin of index $h_k(v(X_t))$.
                \State Update $I_{v(X_t)} \leftarrow 0.$
            \Else
                \State Pack $X_t$ into a new bin.
            \EndIf
            \State $t \leftarrow t + 1$.
        \EndWhile
        \State $k \leftarrow k + 1$.
\EndFor
\end{algorithmic}
\end{algorithm}

Algorithm \ref{alg:overflow} describes our online bin packing algorithm. It solves $K=\ceil{\log_2 T}$ offline bin packing problems throughout the procedure and assumes that an exact optimal solution is available. Later in Section \ref{sec_approx}, we will show that the offline solver can also be replaced by its LP relaxation and an according approximate solution. At the beginning of each phase $k=1,...,K$, we collect all the previously observed item sizes into a set $\mathcal{X}_{T_{k-1}}$ and solve the problem $\text{OPT}(\mathcal{X}_{T_{k-1}}).$ The obtained optimal solution will prescribe our packing decisions until time $T_{k}$, and then a new offline bin packing problem will be solved. Specifically, we order the previously observed item sizes ascendingly by
$$X^{(1)} \leq \dots \leq X^{(T_{k-1})}$$
and, at the same time, index the bins used in the optimal solution of $\text{OPT}(\mathcal{X}_{T_{k-1}}).$ Then we record the packing bin (slot) of each item with the mapping $h_k(\cdot)$ and conclude the offline part of the $k$-th phase. For the online part, at each time $t=T_{k-1}+1,...,T_{k}$, when a new item of size $X_t$ arrives, we resort to the offline packing scheme and find the smallest vacancy index $v(X_t)$ such that (i) $X_t$ is no greater than $X^{(v(X_t))}$ in size and (ii) the slot $h_k(v(X_t))$ is available. When such an index exists, we will pack the item $X_t$ into the slot of item $X^{(v(X_t))}$ in the offline packing scheme and mark that slot as unavailable. Otherwise, if all the slots that can hold $X_t$ are unavailable, we will pack $X_t$ into a new bin. 

The algorithm works regardless of the underlying model and it is a ``learning-based'' algorithm in that it utilizes the history observations for future decision making, though it does not perform an explicit fitting of the underlying distribution. The history observations are utilized in a non-parametric way and connect with the future decisions through the offline optimal packing scheme. In the rest of this section, we will develop a symmetric view between the history and future observation and use this symmetricity to derive the first analytical result for the algorithm. 

\subsection{Exchangeable Random Variables}

We detour a bit for the moment and introduce the notion of exchangeable random variables, also known as exchangeable sequence of random variables. 

\begin{definition}[Exchangeable Random Variables]
For any sequence of random variables $\{\xi_1, \xi_2, \dots\}$ which may be finitely or infinitely long, it is called exchangeable if its joint distribution does not change when finitely many of the random variables are permuted.
\end{definition}

Under the definition, if $\sigma$ denotes a finite permutation of the indices $1,2,3,...$ (i.e., only finitely many indices change positions), then the joint distribution of the permuted sequence $\{\xi_{\sigma(1)}, \xi_{\sigma(2)}, ... \}$ is the same as the sequence $\{\xi_{1}, \xi_{2}, ... \}$. The main result that we will use in our paper is the Maurey-Pisier maximal inequality of exchangeable random variables \citep{maureyannexe}. Here we adopt the presentation in \cite{chobanyan2001exact} as the following proposition.

\begin{proposition}[Theorem 1.1 in \cite{chobanyan2001exact}]
\label{prop:exchange}
Let $\{\xi_1, \xi_2, \dots, \xi_n\}$ be a finite sequence of exchangeable random variables taking real values with $\sum_{i=1}^n \xi_i = 0$. Then the following two-sided inequality holds:
$$\frac12 \E\left[\left\vert\sum_{i=1}^n \xi_i r_i\right\vert\right] \leq \E\left[\max_{k \leq n}\left\vert\sum_{i=1}^k \xi_i\right\vert\right] \leq 2\E\left[\left\vert\sum_{i=1}^n \xi_i r_i\right\vert\right],$$
where $\{r_i\}$ is a sequence of Rademacher random variables independent of $\{\xi_i\}.$ Specifically, $\prob(r_i=1) = \prob(r_i=-1) =\frac{1}{2}.$
\end{proposition}

The middle term in the above inequality concerns the maximal partial sum of the random variable sequence. Specifically, the inequality bounds the maximal partial sum by a more tractable form, that is, the sum of the whole random variable sequence weighted by an independent sequence of Rademacher random variables. The inequality is known as the maximal inequality for exchangeable random variables, and it can be interpreted as a discretized version (though more complicated to derive) of the maximal inequality of a Brownian bridge which can be computed in closed-form (See \cite{durrett2019probability}). The inequality also shares the similar intuition with the symmetrization arguments in empirical process \citep{pollard2012convergence} and statistical learning theory \citep{gyorfi2002distribution}.

Compared to an i.i.d. sequence of random variables, the advantage of exchangeable random variables is that it allows some dependency between the underlying random variables as long as they preserve some permutation symmetry. As we will see, this flexibility of the dependency brings much convenience to the analysis of the packing dynamic.

\subsection{A Warm-up Analysis of the Packing Dynamic}

\label{sec_warm_up}

In this subsection, we single out one phase of the algorithm and analyze the packing dynamic, namely, the interplay between the offline optimal solution and the online packing procedure. Let $\mathcal{X}$ denote the set of items in the offline problem (observed in the previous phases) and let $\mathcal{X}'$ denote the set of items encountered in the online procedure of the current phase. In the context of Algorithm \ref{alg:overflow}, $\mathcal{X}$ and $\mathcal{X}'$  represent the sets $\mathcal{X}_{T_{k-1}}$ and $\mathcal{X}_{(T_{k-1}+1):T_k}$, respectively. It does not hurt to assume $|\mathcal{X}| = |\mathcal{X}'|=N$ now for some $N\in\mathbb{N}$. In fact, the cardinalities of these two set may differ up to $1$ while implementing the algorithm, but this makes no essential change to the analysis. Under both the stochastic model and the random permutation model, the distributions of $\mathcal{X}$ and $\mathcal{X}'$ are the same, 
$$\mathcal{X} \overset{\mathcal{D}}{=} \mathcal{X}'.$$
Next, we place all the item sizes in $\mathcal{X}\cup \mathcal{X}'$ upon the interval $[0,1]$ as in Figure \ref{Fig:packing}.

\begin{figure}[ht!]
\centering
\includegraphics[scale=0.4]{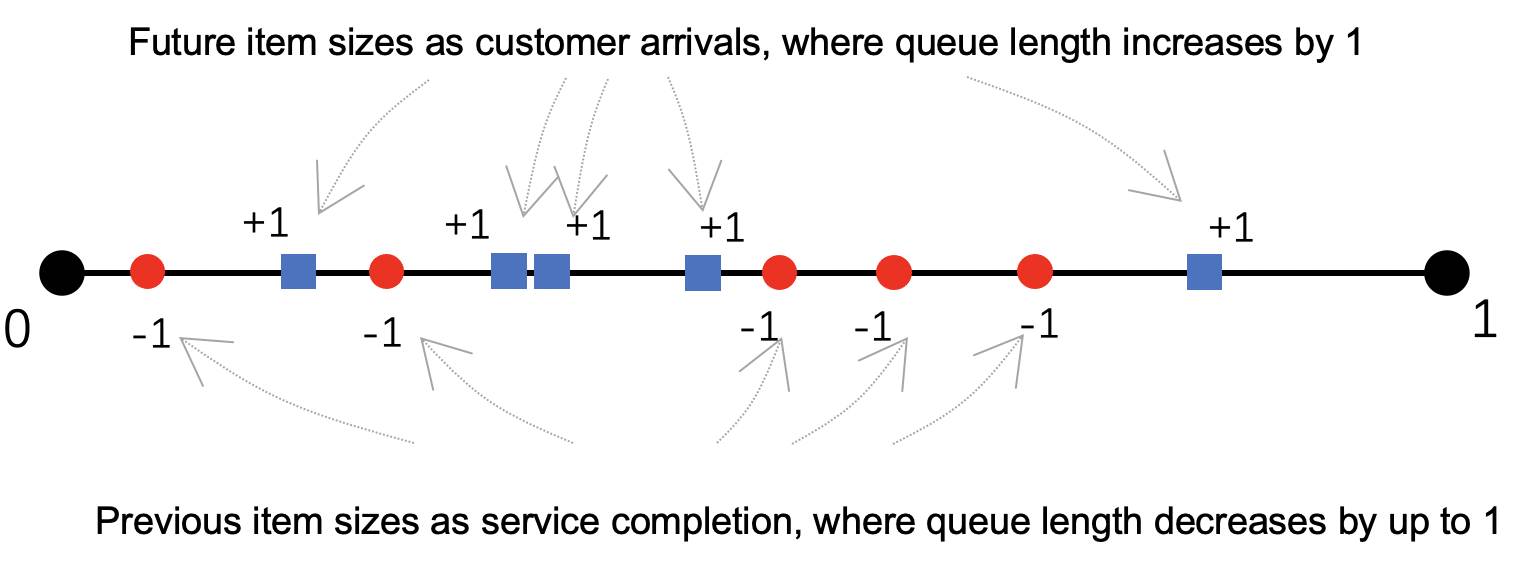}
\caption{Illustration of the packing dynamic over one phase}
\label{Fig:packing}
\end{figure}

Specifically, if there are duplicated item sizes in the set, we randomly permute them to decide their orders on the interval. In this way, all the item sizes from $\mathcal{X}\cup \mathcal{X}'$ are placed on the interval with an absolute order. Then, we mark a ``$-1$'' for an item size from $\mathcal{X}$ and a ``$+1$'' for an item size from $\mathcal{X}'$. 

To better understand the packing dynamic of the phase, we consider a single-server queue defined by traversing from $0$ to $1$ throughout the interval. On a high level, all the items in $\mathcal{X}'$ (the online phase) are viewed as arriving ``customers'' to be served and all the items in  $\mathcal{X}$ (the offline optimal solution) are viewed as completing ``services''. The rule of Algorithm \ref{alg:overflow} is that the arriving customers (items in $\mathcal{X}'$) can only be served by the completed services (items in $\mathcal{X}$) that have a no smaller size. Specifically, we treat all the $-1$'s as service tokens and all the $+1$'s as arrival tokens, namely, the time that a service is completed or a new customer arrives. Each $+1$ represents a new item to be packed, and it can only be packed (served) by a $-1$ to its left on the interval; the reason for the left is that an item can fit in a slot (in the offline solution) that holds an item greater or equal to its size. Mathematically, there are $2N$ tokens on the interval, and each token is associated with a $-1$ or $+1.$ We denote these $2N$ numbers as $\xi_1,...,\xi_{2N}$, the locations of which span from left to right on the interval. A queue length process can be defined by 
$$Q_0=0, \ \ Q_{n} = \max\{0, Q_{n-1}+\xi_n\}$$
for $n=1,...,2N.$ To see the queue dynamic, when a new item arrives to be packed, it joins the queue to wait for the next available slot. When a -1 is encountered, the queue length decreases by one if the current queue length is positive; if the current queue length is zero, the slot in the offline solution corresponding to this -1 is wasted, and the queue length remains to be zero. The queue length process $Q_n$'s describes the packing dynamic of Line 10-21 in Algorithm \ref{alg:overflow}. Importantly, it is easy to verify that the online packing routine in phase $k$ will use 
$$\text{OPT}(\mathcal{X}) + Q_{2N}$$ many bins where $\mathcal{X}$ represents $\mathcal{X}_{T_{k-1}}$ and $N=|\mathcal{X}|=T_{k-1}.$ 

Now we turn our attention to the joint distribution of $\xi_1,...,\xi_{2N}$. Recall that $\mathcal{X}\overset{\mathcal{D}}{=}\mathcal{X}'.$ There are $N$ of $+1$ and $N$ of $-1$ among these $2N$ random variables; as a result, 
$$\sum_{n=1}^{2N} \xi_n = 0.$$
Furthermore, applying the symmetry argument again for $\mathcal{X}$ and $\mathcal{X}',$ we know that the joint distribution of $\xi_1,...,\xi_{2N}$ is that of randomly selecting $N$ indices out of $\{1,2,...,2N\}$ and letting them to be $+1$ with the remaining to be $-1$. Therefore, the joint distribution is invariant under a finite random permutation, and thus $\xi_1,...,\xi_{2N}$ is a sequence of exchangeable random variables. Then we can apply Proposition \ref{prop:exchange} and yield the following bound.

\begin{proposition}
The following inequality holds for any $N\in\mathbb{N},$
$$\E[Q_{2N}] \le 2\sqrt{2N}.$$
\label{prop_queue}
\end{proposition}

\begin{proof}
It is easy to verify that the final queue length equals to the negative of maximal partial sum, i.e.,
$$Q_{2N} = \max_{1\le n' \le 2N} -\sum_{n=1}^{n'} \xi_{n}$$
holds for all $N\in \mathbb{N}$.

Then, we have
\begin{align*} 
\E[Q_{2N}] & = \E\left[\max_{1\le n' \le 2N} -\sum_{n=1}^{n'} \xi_{n}\right] \\ 
   & \leq 
   \E\left[\max_{1\le n' \le 2N} \left\vert\sum_{n=1}^{n'} \xi_{n}\right\vert\right]\\
    & \leq  2 \E\left[\left|\sum_{n=1}^{2N} \xi_n r_n\right| \right]\\
    & \leq  2\sqrt {\E\left[\left|\sum_{n=1}^{2N} \xi_n r_n\right|^2 \right]}
    = 2 \sqrt{2N}
\end{align*}
where $r_n$'s are i.i.d. Rademacher random variables as in Proposition \ref{prop:exchange}. Here the second line is trivial, the third line applies Proposition \ref{prop:exchange}, and the last line applies the Cauchy-Schwartz inequality and uses the independence of $r_n$'s.
\end{proof}

If we apply the above proposition for each phase of the algorithm, we can obtain the following corollary. Its proof simply treats the packing dynamic of each phase as an individual queueing process and applies Proposition \ref{prop_queue}. The extra one accounts for the bin used in phase $0.$ The corollary provides a bound on the expected online bin usage by the offline optimal objectives. 

\begin{corollary}
\label{coro_on2off}
For each phase $k=1,...,K$, Algorithm \ref{alg:overflow} uses no more than
$$\E\left[\mathrm{OPT}(\mathcal{X}_{T_{k-1}})\right]+2\sqrt{2T_{k-1}}$$
many bins on expectation. Thus the online algorithm uses in total no more than $$\sum_{k=1}^K \left(\E\left[\mathrm{OPT}(\mathcal{X}_{T_{k-1}})\right]+2\sqrt{2T_{k-1}}\right) + 1$$
many bins on expectation. The results hold for both the stochastic model and the random permutation model. 
\end{corollary}

Proposition \ref{prop_queue} and Corollary \ref{coro_on2off} demonstrate our first usage of the exchangeable random variables and the related maximal inequality. The result reduces the analysis of the algorithm's performance to the analysis of a sequence of offline optimal objective values. As noted in the corollary, our analysis so far holds for both the stochastic model and the random permutation model because it only relies on a symmetry between the past and future observations. In the following two sections, our discussion will bifurcate into two separate streams for the two models. The focus for both streams will be on the study of the offline optimal objective values, and surprisingly, we will use the argument of exchangeable random variables again but in a more involved way.

\section{Analysis for Stochastic Model}
\label{sec_stoc}

In this section, we detach from the discussion on the algorithm and introduce a quantile-based certainty-equivalent benchmark for the bin packing problem under the stochastic model. We first establish the relation between the benchmark and the offline optimal solution, and then show an equivalence between this new benchmark and the conventional asymptotic benchmark for the problem. Finally, we conclude the section with a regret bound for Algorithm \ref{alg:overflow} under the stochastic model. 

\subsection{Quantile-Based Certainty-Equivalent Benchmark}

Recall that $F$ denotes the item size distribution. With a bit abuse of the notation, we use $F$ to refer to both the distribution itself and the cumulative distribution function of $X_t$'s.

\begin{definition}[Quantile]
For a distribution $F$ supported on $[0,1]$, define its $\alpha$-quantile as
$$ F^{-1}(\alpha) \coloneqq \inf \{y | F(y) \geq \alpha\}.$$
In addition, we define $F^{-1}(0)=0$ and $F^{-1}(1)=1$.
\end{definition}

The quantile-based certainty-equivalent benchmark is defined by an offline (deterministic) bin packing problem with item sizes being the quantiles of $F$. We do not assume the existence of the limit for now and will show in the next subsection that the limit indeed exists by Kingman's subadditive argument for the asymptotic benchmark.

\begin{definition}[Quantile-Based Certainty-Equivalent Benchmark]
For a distribution $F$ supported on $[0,1]$, define
\begin{align*}
    \mathrm{CE}(F) & \coloneqq  \limsup_{T\rightarrow \infty} \frac1T \mathrm{OPT}\left(\left\{F^{-1}\left(\frac{0}T\right),\dots,F^{-1}\left(\frac{T-1}T\right)\right\}\right)\\
    & =  \limsup_{T\rightarrow \infty} \frac1T \mathrm{OPT}\left(\left\{F^{-1}\left(\frac{1}T\right),\dots,F^{-1}\left(\frac{T}T\right)\right\}\right).
\end{align*}
\end{definition}

To make some intuition for the benchmark definition, the set of quantiles provide a good characterization of the distribution $F$. Moreover, the two optimization problems on the right-hand-side in above are deterministic. The following lemma can be easily proved by comparing the optimal solutions to the two bin packing problems. It validates the two equivalent definitions of the benchmark in above.

\begin{lemma}It holds for all distribution $F$ and $T\in \mathbb{N}^+$ that
\begin{align*}
   \mathrm{OPT}\left(\left\{F^{-1}\left(\frac{0}T\right),\dots, F^{-1}\left(\frac{T-1}T\right)\right\}\right) & \le \mathrm{OPT}\left(\left\{F^{-1}\left(\frac{1}T\right),\dots, F^{-1}\left(\frac{T}T\right)\right\}\right) \\
   & \le \mathrm{OPT}\left(\left\{F^{-1}\left(\frac{0}T\right),\dots, F^{-1}\left(\frac{T-1}T\right)\right\}\right)+1.
\end{align*}
\end{lemma}

Now we proceed to establish the relation between $\text{CE}(F)$ and $OPT(\mathcal{X}_T).$ First, the following lemma tells that the offline optimal value is upper bounded by the deterministic quantile problem plus $O(\sqrt{T})$. The proof is based on a similar symmetrization argument as the last section that we analyze the dynamic of packing $\mathcal{X}_T$ into the optimal solution of the bin packing problem on the right-hand-side. 

\begin{proposition}
\label{prop:certainty}
Suppose $\mathcal{X}_T=\{X_1, \dots, X_T\}$ consists of i.i.d. samples generated from $F$. Then, the following inequality holds for all distribution $F$ and $T\in \mathbb{N}^+$,
$$ \E[\mathrm{OPT}(\mathcal{X}_T)] \leq \mathrm{OPT}\left(\left\{F^{-1}\left(\frac{1}T\right),F^{-1}\left(\frac{2}T\right), \dots, F^{-1}\left(\frac{T}T\right)\right\}\right) + 2\sqrt{T}. $$
\end{proposition}
\begin{proof}
For simplicity, we prove the result for a continuous distribution $F$ and will explain afterwards on how the proof works for a general distribution. We make a ``constructive'' proof through packing $\mathcal{X}_T$ into the optimal solution of the right-hand-side. Specifically, we mimic Line 6-21 of Algorithm \ref{alg:overflow} and solve an offline problem with item sizes $\left\{F^{-1}\left(\frac{1}T\right), \dots, F^{-1}\left(\frac{T}T\right)\right\}$. Then, with the same routine of the algorithm, we pack $\mathcal{X}_T$ into the optimal offline solution generated from $\mathrm{OPT}\left(\left\{F^{-1}\left(\frac{1}T\right),F^{-1}\left(\frac{2}T\right), \dots, F^{-1}\left(\frac{T}T\right)\right\}\right).$

The dynamic in Section \ref{sec_warm_up} then becomes
$$ Q_{0} = 0, \ \ Q_{t} =  \max\{Q_{t-1} + \psi_{t} - 1, 0\},$$
where
$$ (\psi_{1},\dots,\psi_{T}) \sim \mathrm{Multinomial}\left(T, \left(\frac{1}{T},\dots,\frac{1}{T}\right)\right). $$
In fact, this can be seen from the fact that
when the underlying distribution is continuous, we have $$\prob\left(X_t\in \left(\frac{t'-1}{T}, \frac{t'}{T}\right)\right) = \frac{1}{T}$$
for all $t$ and $t'$. Thus, the items are multinomially distributed among all the intervals generated by the quantiles
$\left\{\left(0,F^{-1}\left(\frac{1}T\right)\right], \dots, \left(F^{-1}\left(\frac{T-1}T\right), F^{-1}\left(\frac{T}T\right)\right]\right\}.$
At the end of the $t$-th interval, there is a server of size $F^{-1}\left(\frac{t}{T}\right)$, which serves at most $1$ customer.

As in Section \ref{sec_warm_up}, we can construct a new sequence by discarding the non-negativity part,
$$S_{0} \coloneqq 0, \ \  S_{t} \coloneqq S_{t-1} + \psi_{t} - 1.$$
Then the queue length $Q_{T}$ is given by
$$Q_{T} = -\min_{t\leq T}\left\{S_t\right\}.$$
Note that all $\psi_t - 1$ s are exchangeable, since their corresponding probabilities are all equal.\\
Furthermore, 
$$\sum_{t=1}^{T} (\psi_t - 1) = T - T = 0.$$
Thus by Proposition \ref{prop:exchange},
\begin{align*}
    \E[Q_T] & = \E\left[\left|\min_{t \leq T} S_t\right|\right]\\
    & \leq 2 \E\left[\left|\sum_{t=1}^T (\psi_t - 1) \cdot r_t\right|\right]\\
    & \leq 2 \left(\E\left[\left|\sum_{t=1}^T (\psi_t - 1)^2 r_t^2\right|\right]\right)^{\frac{1}2} \\
    & = 2 \left(\sum_{t=1}^T \left(1 - \frac{1}{T}\right)\right)^{\frac{1}2} \leq 2\sqrt{T},
\end{align*}
where $r_t$'s are Rademacher random variables independent of $\psi_t$'s.

Now we argue how the analysis also works for a general distribution $F$. First, note that the only place we utilize the continuous property of the distribution $F$ is that $\prob\left(X_t\in \left(\frac{t'-1}{T}, \frac{t'}{T}\right)\right) = \frac{1}{T}$
for all $t$ and $t'$. For a general distribution, this is violated, but in a favorable manner (that the probability of an item capable of being packed to each corresponding slot is larger than $1/T$). Specifically, if we construct the queue in the same way as above and denote the corresponding increments as $(\psi_{1}',\dots,\psi_{T}')$.
By the definition of quantile, we know 
$$(\psi_{1}',\dots,\psi_{T}') 	\succeq (\psi_{1},\dots,\psi_{T}) $$
where $\succeq$ denotes a pairwise stochastic dominance. Therefore the expected queue length $Q_T$ of a continuous distribution will be an upper bound for that of a general distribution. Thus we complete the proof.
\end{proof}

Also, the inequality is tight in terms of the order on $T$. To see this, let the distribution $F$ be $ \prob(X = 1/2 + \epsilon) = \prob(X = 1/2 - \epsilon) = 1/2$ for some $\epsilon\in(0,1/2)$. Then $ \mathrm{CE}(F) = 1/2.$ However,
\begin{align*}
    \E[\mathrm{OPT}(\mathcal{X})] = & \frac{T}2   + \frac12 \E\left[\max\left\{ \sum_{t=1}^T r_t, 0\right\}\right]\\
    = & \frac{T}2    + \Omega(\sqrt{T}),
\end{align*}
where $r_t$'s are independent Rademacher random variables.

The following corollary establishes the relation between $\text{OPT}(\mathcal{X}_T)$ and $\text{CE}(F)$ using Proposition \ref{prop:certainty}. In its proof, it uses the relaxed LP benchmark OPT$_f$ as a proxy for that the convexity property of the relaxed problem better relates optimal objective values with respect to different $T$'s. 

\begin{corollary}
[CE as upper bound of OPT]
\label{coro:ce}
Suppose $\mathcal{X}_T=\{X_1,\dots,X_T\}$ are i.i.d. samples from distribution $F$.
Then
$$ \sup_{\mathrm{supp}(F)\subset[0,1]}\{ \E[\mathrm{OPT}(\mathcal{X}_T)] - T\cdot \mathrm{CE}(F)\} = \Theta (\sqrt{T}). $$
\end{corollary}

\begin{proof}
By Proposition \ref{prop:relaxedgap},
$$ \mathrm{OPT}_f(\mathcal{X}_T) \leq \mathrm{OPT}(\mathcal{X}_T) \leq \mathrm{OPT}_f(\mathcal{X}_T) + 4 \log_2^2 T + 17 \log_2 T + 11.$$

Thus,
\begin{align*}
    \E[\mathrm{OPT}(\mathcal{X}_T)] & \leq  \mathrm{OPT}\left(\left\{F^{-1}\left(\frac{0}T\right),\dots,F^{-1}\left(\frac{T-1}T\right)\right\}\right) + 2\sqrt{2T} + 1\\
    & \leq  \mathrm{OPT}_f\left(\left\{F^{-1}\left(\frac{0}T\right),\dots,F^{-1}\left(\frac{T-1}T\right)\right\}\right) + 2\sqrt{2T} + 4 \log_2^2 T + 17 \log_2 T + 12\\
    & =  \frac12 \mathrm{OPT}_f\left(2\left\{F^{-1}\left(\frac{0}T\right),\dots,F^{-1}\left(\frac{T-1}T\right)\right\}\right) + 2\sqrt{2T} + 4 \log_2^2 T + 17 \log_2 T + 12\\
    & \leq \frac12 \mathrm{OPT}_f\left(\left\{F^{-1}\left(\frac{0}{2T}\right),\dots,F^{-1}\left(\frac{2T-1}{2T}\right)\right\}\right) + 2\sqrt{2T} + 4 \log_2^2 T + 17 \log_2 T + 12
\end{align*}
where the first line is from Proposition \ref{prop:certainty}, the second line is from Proposition \ref{prop:relaxedgap}, and the third and fourth lines are simply from playing with the feasibility and the constraints of the relaxed LP problem $\text{OPT}_f.$

By repeating the procedure for the first term in above, it implies that
\begin{align*}
\E[\mathrm{OPT}(\mathcal{X}_T)] & \le T\cdot \limsup_{k\rightarrow \infty} \mathrm{OPT}_f\left(\left\{F^{-1}\left(\frac{0}{2^kT}\right),\dots,F^{-1}\left(\frac{2^kT-1}{2^kT}\right)\right\}\right) \\ 
& \ \ \ \ + 2\sqrt{2T} + 4 \log_2^2 T + 17 \log_2 T + 12 \\ & \leq T\cdot \mathrm{CE}(F) + 2\sqrt{2T} + 4 \log_2^2 T + 17 \log_2 T + 12.
\end{align*}
\end{proof}

The following proposition establishes the $\text{CE}(F)$ as a lower bound for the OPT. Putting together Corollary \ref{coro:ce} with Proposition \ref{prop:celower}, we conclude that $T\cdot \mathrm{CE}(F)$ well approximates $\E[\text{OPT}(\mathcal{X}_{T})]$ with a gap of $\Theta(\sqrt{T}).$

\begin{proposition}[CE as lower bound of OPT]
\label{prop:celower}
For $\mathcal{X}_T = \{X_1,\dots,X_T\}$ i.i.d. sampled from $F$, we have
$$ T\cdot \mathrm{CE}(F) - \E[\mathrm{OPT}(\mathcal{X}_T)] \le 4 \log_2^2 T + 17 \log_2 T + 12. $$
\end{proposition}
\begin{proof}
For distribution $F$, we define the following distribution through its cumulative density function
$$\underline{F}_T(x) = \sum_{t=1}^T I\left(F^{-1}\left(\frac{t-1}{T}\right)\le x\right)$$
for $x\in[0,1].$ Here $I(\cdot)$ is the indicator function. To make some intuition, when $F$ is a continuous distribution, $\underline{F}_T$ is a uniform distribution on the set $\{F^{-1}(0/T),...,F^{-1}(T-1/T)\}.$

Let $\underline{\mathcal{X}}_T = \{\underline{X}_1,...,\underline{X}_T\}$ be a set of i.i.d. samples from the distribution $\underline{F}_T$. 

Then we have
\begin{align*}
    & \frac1T \left[\mathrm{OPT}\left(\left\{F^{-1}\left(\frac{0}T\right),\dots,F^{-1}\left(\frac{T-1}T\right)\right\}\right) - (4\log_2^2 T + 17 \log_2 T + 11)\right]\\
    \leq & \frac1T \mathrm{OPT}_f\left(\left\{F^{-1}\left(\frac{0}T\right),\dots,F^{-1}\left(\frac{T-1}T\right)\right\}\right)\\
    \leq & \E\left[\frac1T \mathrm{OPT}_f(\underline{\mathcal{X}}_T)\right]\\
    \leq & \E\left[\frac1T \mathrm{OPT}_f(\mathcal{X}_T)\right]\\
    \leq & \E\left[\frac1T \mathrm{OPT}(\mathcal{X}_T)\right].
\end{align*}
Here the first inequality comes from Proposition \ref{prop:relaxedgap}, the second inequality comes from the convexity of linear programs, and the last inequality comes from the relaxation of the integer LP.

Note the fact that $\underline{F}_T(x) \ge F(x)$ for all $x\in[0,1]$. The third inequality can be proved by a coupling argument for the two distributions $\underline{F}_T$ and $F$ such that $\prob(\underline{\mathcal{X}}_T\le \mathcal{X}_T)=1$ where the inequality holds element-wise.

Therefore, we have
$$ \E[\mathrm{OPT}(\mathcal{X})]+4\log_2^2 T + 17 \log_2 T + 11 \geq \mathrm{OPT}\left(\left\{F^{-1}\left(\frac{0}T\right), \dots, F^{-1}\left(\frac{T-1}T\right)\right\}\right). $$
Thus, for any 
$T^\prime = mT + r\ (0\leq r < T)$,
\begin{align*}
    &\E[\mathrm{OPT}(\mathcal{X})]+4\log_2^2 T + 17 \log_2 T + 11 \\
    \geq & \mathrm{OPT}\left(\left\{F^{-1}\left(\frac{0}T\right), \dots, F^{-1}\left(\frac{T-1}T\right)\right\}\right)\\
    \ge & \mathrm{OPT}\left(\left\{F^{-1}\left(\frac{1}T\right), \dots, F^{-1}\left(\frac{T}T\right)\right\}\right) - 1\\
    \geq & \frac{1}{m} \mathrm{OPT}\left(m\left\{F^{-1}\left(\frac{1}T\right), \dots, F^{-1}\left(\frac{T}T\right)\right\}\right) - 1\\
    \geq & \frac{1}{m} \mathrm{OPT}\left(\left\{F^{-1}\left(\frac{1}{mT}\right), \dots, F^{-1}\left(\frac{mT}{mT}\right)\right\}\right) - 1\\
    \geq & \frac{1}{m} \mathrm{OPT}\left(\left\{F^{-1}\left(\frac{1}{mT+r}\right), \dots, F^{-1}\left(\frac{mT}{mT+r}\right)\right\}\right) - 1\\
    \geq & \frac{1}{m} \left[\mathrm{OPT}\left(\left\{F^{-1}\left(\frac{1}{mT+r}\right), \dots, F^{-1}\left(\frac{mT+r}{mT+r}\right)\right\}\right) - r\right] - 1\\
    \geq & \frac{T}{T^\prime}\left[\mathrm{OPT}\left(\left\{F^{-1}\left(\frac{1}{T^\prime}\right), \dots, F^{-1}\left(\frac{T^\prime}{T^\prime}\right)\right\}\right) - r\right] - 1.
\end{align*}
Taking limit with respect to $T^\prime$, we complete the proof.

\end{proof}

\subsection{Relation with the Asymptotic Benchmark}

The quantile-based CE benchmark discussed in the last section works as a proper approximation for the offline optimal objective. We now establish an equivalence between the CE benchmark with the more conventional asymptotic benchmark of the bin packing problem.

\begin{definition}[Asymptotic Benchmark]
For any distribution $F$ supported on $[0,1]$, we define the asymptotic performance benchmark as follows:
$$ \mathrm{AS}(F) \coloneqq \limsup_{T\rightarrow \infty} \E\left[\frac1T \mathrm{OPT}(\mathcal{X}_T)\right]$$
where $\mathcal{X}_T=\{X_1,...,X_T\}$ is a set of i.i.d. samples from $F.$
\end{definition}
\begin{proposition}
\label{prop:CEAS}
We have 
$$ \mathrm{CE}(F) = \mathrm{AS}(F). $$
Moreover, if 
$\lim_{T\rightarrow \infty} \E\left[\frac1T \mathrm{OPT}(\mathcal{X}_T)\right]$
exists, then the limit superior in both definitions of $\mathrm{CE}(F)$ and $\mathrm{AS}(F)$ can be replaced by limit.
\end{proposition}
\begin{proof}
From Corollary \ref{coro:ce} and Proposition \ref{prop:celower}, we know
\begin{align*}
    & \frac1T \left[\mathrm{OPT}\left(\left\{F^{-1}\left(\frac{0}T\right),\dots,F^{-1}\left(\frac{T-1}T\right)\right\}\right) - (4\log_2^2 T + 17 \log_2 T + 11)\right]\\
    \leq & \E\left[\frac1T \mathrm{OPT}(\mathcal{X}_T)\right]\\
    \leq & \frac1T \left[\mathrm{OPT}\left(\left\{F^{-1}\left(\frac{1}T\right),\dots,F^{-1}\left(\frac{T}T\right)\right\}\right) + 2\sqrt{T}\right].
\end{align*}

The proof can be completed by taking limit superior (or limit) on both sides of the inequality.
\end{proof}
Note for two sets $\mathcal{X}, \mathcal{X}'$ of items, 
$$\mathrm{OPT}(\mathcal{X}) + \mathrm{OPT}(\mathcal{X}') \geq \mathrm{OPT}(\mathcal{X} \cup \mathcal{X}').$$
We can then employ the subadditive process argument in \cite{kingman1976subadditive} and validates the existence of the limit $\E[\mathrm{OPT}(\mathcal{X}_{T})/T]$ (See \citep{rhee1988optimal1}). Thus we can remove all the limit superior in this section and conclude with the following proposition.

\begin{proposition} We have
$$ \lim_{T\rightarrow \infty} \frac{1}T \mathrm{OPT}\left(\left\{F^{-1}\left(\frac{1}T\right),\dots,F^{-1}\left(\frac{T}T\right)\right\}\right) $$
exists and it is equal to
$$ \lim_{T \rightarrow \infty} \frac{1}{T} \E\left[\mathrm{OPT}(\mathcal{X}_{T})\right]. $$
\end{proposition}

The results so far in this section concern the so-called stochastic bin packing problem, which are of independent interest. As noted in the introduction, our treatment simplifies and improves the analyses in \citep{rhee1988optimal1, rhee1989optimal2, rhee1989optimal3}. The key of our analysis is to show that the relation between the quantile-based benchmark $\mathrm{OPT}\left(\left\{F^{-1}\left(\frac{1}T\right),\dots,F^{-1}\left(\frac{T}T\right)\right\}\right)$ and OPT$(\mathcal{X}_T)$ remains invariant with respect to different distributions.

\subsection{Regret Analysis for Stochastic Model}

From the analysis in the previous subsection, we know that the benchmark $\text{CE}(F)$ can be used to relate the offline optimal objectives with different horizon lengths. Consequently, together with the analysis in Section \ref{sec_warm_up}, we can derive a regret upper bound for Algorithm \ref{alg:overflow} as follows.

\begin{theorem}[Upper bound for stochastic model]
\label{thm:stochastic}
Suppose $\mathcal{X}_T=\{X_1,\dots,X_T\}$ is a set of i.i.d. samples from distribution $F$. The following inequality holds for all $T\ge 1,$
$$ \sup_{\mathrm{supp}(F)\subset [0,1]} \{\E [N^\pi(\mathcal{X}_T) ]- \E[\mathrm{OPT}(\mathcal{X}_T)]\} \leq 10\sqrt{T} + 2\lceil\log_2 T\rceil^3 + 13\lceil\log_2 T\rceil^2 + 43\lceil\log_2 T\rceil + 13$$
where the expectation is taken with respect to the distribution $F$ and $\pi$ denotes the online packing scheme specified by Algorithm \ref{alg:overflow}.
\end{theorem}
\begin{proof}
From Corollary \ref{coro_on2off} and Corollary \ref{coro:ce}, we know that for each phase $1 \leq k \leq K = \lceil \log_2 T \rceil$,
\begin{align*}
    \E[N^{\pi}(\mathcal{X}_{T_{k-1}+1:T_{k}})] \leq & \E[\mathrm{OPT}(\mathcal{X}_{T_{k-1}})] + 2\sqrt{T_{k-1}}\\
    \leq & T_{k-1}\cdot \mathrm{CE}(F) + 4\sqrt{T_{k-1}} + 4(k-1)^2 + 17(k-1) + 11\\
    \leq & (T_k-T_{k-1})\cdot \mathrm{CE}(F) + 4\sqrt{T_{k-1}} + 4(k-1)^2 + 17(k-1) + 12\\
    \leq & (T_k-T_{k-1})\cdot \mathrm{CE}(\mathcal{F}) + 4\sqrt{T_k-T_{k-1}} + 4(k-1)^2 + 17(k-1) + 16.
\end{align*}
Here, we note 
\begin{align*}
    4\sqrt{T_k - T_{k-1}} \leq  4\sqrt{\frac{T_k}2} 
    \leq 4\sqrt{\frac{T_{k+1} - T_{k}}{2}} + \frac{4}{\sqrt{2}} \le \cdots
    \leq  4\sqrt{\frac{T_{K} - T_{K-1}}{2^{K-k}}} + 10.
\end{align*}
Taking a summation with respect to $k$,
\begin{align*}
    \E[N^{\pi}(\mathcal{X}_{T})] = & 1 + \sum_{k=1}^{K} \E[N^{\pi}(\mathcal{X}_{T_{k-1}+1:T_{k}})]\\
    \leq & T\cdot \mathrm{CE}(F) + \sum_{k=1}^K [4\sqrt{T_k - T_{k-1}} + 4(k-1)^2 + 17(k-1) + 16]\\
    \leq & T\cdot \mathrm{CE}(F) + \sum_{k=1}^K [4\sqrt{T_K - T_{K-1}}\cdot \frac{1}{\sqrt{2}^{K-k}} + 4(k-1)^2 + 17(k-1) + 26]\\
    \leq & T\cdot \mathrm{CE}(F) + \frac{4}{\sqrt{2}-1}\sqrt{T} + \frac{4}{3}K^3 + \frac{17}{2}K^2 + 26K\\
    \leq & T\cdot \mathrm{CE}(F) + 10\sqrt{T} + 2\lceil\log_2 T\rceil^3 + 9\lceil\log_2 T\rceil^2 + 26\lceil\log_2 T\rceil.
\end{align*}
By plugging in the result from Proposition \ref{prop:celower}, we complete the proof.
\end{proof}

\section{Random Permutation Model and Approximate Algorithms}

\label{sec_RP_AA}

\subsection{Random Permutation Model}

\label{secRandPerm}

The random permutation model can be viewed as a sampling without replacement procedure from a ground set $\mathcal{X}_T' = \{X'_1,...,X'_T\}$, while item sizes $X_1',...,X_T'$ can be adversarially chosen. In this section, we show that compared to the stochastic model, the random permutation model does not make essential change to the nature of the problem. As the random permutation model can be viewed as a more general model than the stochastic model, the analysis in this section serves as an alternative proof for the regret bound in the last section. Without loss of generality, we assume the ground set $\mathcal{X}_T'$ is ordered, i.e.,
$$0<X_1' \le X_2' \le \cdots \le X_T' \le 1.$$
We denote the arrival sequence of the items as $\mathcal{X}_T=\{X_1,...,X_T\}$ where $X_t$ is the $t$-th item arrived in the online sequence, i.e., $t$-th item sampled from $\mathcal{X}_T'$ without replacement.

The following proposition establishes the relation between the optimal objective value of a randomly selected subset of items and that of the whole set. We remark that the factor of $2^k$ is not essential and it can be replaced by any positive integer, but the current choice suffices for our analysis.

\begin{proposition}
\label{prop:RP}
For a fixed $k\in \mathbb{N}$, let $\tau = \ceil{\frac{T}{2^k}}$ and let $\mathcal{X}_{\tau} = \{X_1,...,X_{\tau}\}$ be a set of $\tau$ items sampled without replacement from $\mathcal{X}_T'.$ Then we have 
$$ \E[\mathrm{OPT}(\mathcal{X}_{\tau})] \leq \frac1{2^k} \mathrm{OPT}(\mathcal{X}_{T}') + 2\sqrt{\tau} + 4\log_2^2 \tau + 17\log_2 \tau + 13$$
where the expectation is taken with respect to the random permutation/sampling without replacement.
\end{proposition}

\begin{proof}
Let $T_0=2^k \tau \ge T.$ We first append $T_0-T$ items with size one to construct a new set $\mathcal{X}_{T_0}'$, i.e., $\mathcal{X}_{T_0}'=\mathcal{X}_T'\cup\{X_{T+1}',\dots, X_{T_0}'\}$ where $X_{T+1}'=\cdots=X_{T_0}'=1.$ 
The addition of these items will not make no difference to the packing dynamic as we can always re-sample another item if we sampled a newly added item out of the original $\mathcal{X}_T$, and then place the re-sampled item into the slot prepared for the newly added item of size $1$. So we will proceed to analyze the packing dynamic with this $\mathcal{X}_{T_0}$ instead.

Consider the set $\{X_{2^k}', X_{2^k\cdot 2}',\dots,X_{2^k\cdot \tau}'\}$ and the packing of $\mathcal{X}_\tau$ into the optimal packing solution of $$\mathrm{OPT}(\{X_{2^k}',X_{2^k\cdot 2}',\dots,X_{2^k\cdot \tau}'\}).$$

Then we divide $\mathcal{X}_{T_0}'$ into $\tau$ disjoint subsets, where for $i=1,...,\tau$ the $i$-th subset consists of $2^k$ items, namely, $X_{2^k \cdot(i-1) + 1}', \dots, X_{2^k \cdot i}'$. Then the distribution of $\mathcal{X}_{\tau}$ items among all $\tau$ subsets follows the hypergeometric distribution, where each subset contains $2^k$ items. We define the number of items in $\mathcal{X}_{\tau}$ from subset $i$ to be $\eta_i$ and then,
$$(\eta_1,\dots,\eta_{\tau}) \sim \mathrm{Hypergeometric}(\tau, T_0, (2^k,\dots,2^k)). $$
As before, we view each item as a customer and each slot as a server. Then we can define the queue length process
$$ Q_0 \coloneqq 0,\ \ Q_i \coloneqq \max\{Q_{i-1} + \eta_i - 1, 0\}$$
for $i=1,...,\tau$.

Since all $\eta_i$'s are exchangeable with $\sum_{i=1}^{\tau} (\eta_i - 1) = \tau - \tau = 0,$
we have
\begin{align*}
    \E[Q_{\tau}] = & \E\left[\left|\min_{k\leq \tau} \sum_{i=1}^k (\eta_i - 1)\right|\right]\\
    \leq & 2\E\left[\left|\sum_{i=1}^{\tau} (\eta_i - 1)\cdot r_i\right|\right]\\
    \leq & 2\left(\E\left[\left|\sum_{i=1}^{\tau} (\eta_i - 1)^2 r_i^2\right|\right]\right)^{\frac12}\\
    = & 2(\tau\cdot \mathrm{Var}(\eta_i))^{\frac{1}2}\\
    = & 2(\tau\cdot \frac{(\tau - 1)(2^k - 1)}{2^k \cdot \tau - 1})^{\frac{1}2}\\
    \leq & 2\sqrt{\tau}.
\end{align*}
This means
\begin{equation*}
\E[\mathrm{OPT}(\mathcal{X}_\tau)] \le \mathrm{OPT}(\{X_{2^k}',X_{2^k\cdot 2}',\dots,X_{2^k\cdot \tau}'\}) + 2\sqrt{\tau}.
\end{equation*}
Finally, for the first term on the right-hand-side in above, we can used the relaxed optimal objective and have the following
\begin{align*}
    &\mathrm{OPT}(\{X_{2^k}',X_{2^k\cdot 2}',\dots,X_{2^k\cdot \tau}'\}) \\
    \leq & \mathrm{OPT}_f(\{X_{2^k}',X_{2^k\cdot 2}',\dots,X_{2^k\cdot \tau}'\}) + 4\log_2^2 \tau + 17\log_2 \tau + 11\\
    \leq & \frac1{2^k} \mathrm{OPT}_f(\{X_{2^k}',X_{2^k+1}',\dots,X_{2^k\cdot(\tau + 1) - 1}'\}) + 4\log_2^2 \tau + 17\log_2 \tau  + 11\\
    \leq & \frac1{2^k} \mathrm{OPT}_f(\{X_{2^k}',X_{2^k+1}',\dots,X_{T}'\}) +  4\log_2^2 \tau + 17\log_2 \tau + 13\\
    \leq & \frac1{2^k} \mathrm{OPT}_f(\{X_1', X_2',\dots,X_{T}'\}) + 4\log_2^2 \tau + 17\log_2 \tau + 13.
\end{align*}
Here the first line is from Proposition \ref{prop:relaxedgap}. The second line comes from appending after each item $X_{2^k\cdot i}'$ with $2^{k}-1$ items with larger or equal size for $i=1,...,\tau$. The third line comes from by removing the last $2\cdot 2^k$ items and this can be done with at most $2\cdot 2^k$ bins. The last line comes from adding $2^k-1$ items in the front. 
\end{proof}

Proposition \ref{prop:RP} is parallel to Proposition \ref{prop:certainty} and Proposition \ref{prop:celower} in that these results all provide us a way to relate optimal objective values of different horizon lengths. Thus with Proposition \ref{prop:RP}, we can apply the same reasoning as Theorem \ref{thm:stochastic} to obtain the following regret bound for the random permutation model.

\begin{theorem}[Upper bound for random permutation model]
Suppose $\mathcal{X}_T = \{X_1,\dots,X_T\}$ is a randomly permuted sequence of the ground set $\mathcal{X}'_T=\{X_1',\dots,X_T'\}$. 
The following inequality holds for all $T\ge 1,$
$$ \E[N^{\pi}(\mathcal{X}_T)]\leq \mathrm{OPT}(\mathcal{X}_T') + 12\sqrt{T} + 2\lceil\log_2 T\rceil^3 + 9\lceil\log_2 T\rceil^2 + 31\lceil\log_2 T\rceil + 1$$
where the expectation is taken with respect to the random permutation and $\pi$ denotes the online packing scheme specified by Algorithm \ref{alg:overflow}.
\end{theorem}

\begin{proof}
For each phase $1 \leq k \leq K = \lceil \log_2 T \rceil$, $T_k = \lceil \frac{T}{2^{K-k}} \rceil.$ From Corollary \ref{coro_on2off} and Proposition \ref{prop:RP}, we know
\begin{align*}
    \E[N^\pi(\mathcal{X}_{T_{k-1}+1:T_k})] \leq & \E[\mathrm{OPT}(\mathcal{X}_{T_{k-1}})] + 2\sqrt{2T_{k-1}}\\
    \leq & \frac{1}{2^{K-k+1}}\E[\mathrm{OPT}(\mathcal{X}_T')] + (2\sqrt{2} + 2)\sqrt{T_{k-1}} + 4\log_2^2 T_{k-1} + 17\log_2 T_{k-1} + 13\\
    \leq & \frac{1}{2^{K-k+1}}\E[\mathrm{OPT}(\mathcal{X}_T')] + (2\sqrt{2} + 2)\sqrt{T_k - T_{k-1}} + 4(k-1)^2 + 17(k-1) + 19.
\end{align*}
Here the first line comes from Corollary \ref{coro_on2off}, the second line comes from Proposition \ref{prop:RP}, and the last line comes from plugging in the value $T_k$ (and $T_{k-1}$).

Moreover, simple calculation tells
\begin{align*}
    (2\sqrt{2} + 2)\sqrt{T_k - T_{k-1}} \leq (2\sqrt{2} + 2)\sqrt{\frac{T_k}{2}}
    \leq (2\sqrt{2} + 2)\sqrt{\frac{T_K - T_{K-1}}{2^{K-k}}} + 12.
\end{align*}

Sum all above estimations up,
\begin{align*}
    \E[N^\pi(\mathcal{X}_T)] = & 1 + \sum_{k=1}^{K} \E[N^\pi(\mathcal{X}_{T_{k-1}+1:T_k})]\\
    \leq & \E[\mathrm{OPT}(\mathcal{X}_T')] + \sum_{k=1}^K [(2\sqrt{2} + 2)\sqrt{T_k - T_{k-1}} + 4(k-1)^2 + 17(k-1) + 19] + 1\\
    \leq & \E[\mathrm{OPT}(\mathcal{X}_T')] + \sum_{k=1}^K [(2\sqrt{2} + 2)\sqrt{\frac{T_K - T_{K-1}}{2^{K-k}}} + 4(k-1)^2 + 17(k-1) + 31] + 1\\
    \leq & \E[\mathrm{OPT}(\mathcal{X}_T')] + \frac{2\sqrt{2} + 2}{\sqrt{2}-1}\sqrt{T} + \frac{4}{3}K^3 + \frac{17}{2}K^2 + 31K + 1\\
    \leq & \E[\mathrm{OPT}(\mathcal{X}_T')] + 12\sqrt{T} + 2\lceil\log_2 T\rceil^3 + 9\lceil\log_2 T\rceil^2 + 31\lceil\log_2 T\rceil + 1.
\end{align*}
Thus we complete the proof.
\end{proof}

\subsection{Approximate Algorithm as Offline Oracle}

\label{sec_approx}

Now we show that Algorithm \ref{alg:overflow} does not require an exact optimal solution for the problem OPT$(\mathcal{X}_{T_{k-1}})$. Many approximate algorithms have been proposed for solving (offiline) bin packing problem based on the Gilmore-Gomory LP relaxation \citep{gilmore1961linear}, such as \cite{karmarkar1982efficient}, \cite{rothvoss2013approximating}, and \cite{hoberg2017logarithmic}. For example, the well-known grouping algorithms in \cite{karmarkar1982efficient} use the relaxed linear program to produce an approximated solution for the integer program. Their relaxation gap is on the order of $O(\log^2(\mathrm{OPT}_f))$ where the constants are shown in Proposition \ref{prop:relaxedgap} and $\mathrm{OPT}_f$ is the optimal value of the Gilmore-Gomory LP relaxation. Another better approximation is obtained by \cite{hoberg2017logarithmic}. They present a polynomial time algorithm that uses up to $O(\log (\mathrm{OPT}_f))$ additional bins.
\begin{proposition}[Theorem 1.2 in \cite{hoberg2017logarithmic}]
\label{prop:logarithm}
For any bin packing instance with item sizes $X_1, \dots, X_T \in [0,1]$, one can compute a solution with at most $\mathrm{OPT}_f + O(\log (\mathrm{OPT}_f))$ bins, where $\mathrm{OPT}_f$ denotes the optimal value of the Gilmore-Gomory LP relaxation. The algorithm is randomized and the expected running time is polynomial in $T$.
\end{proposition}


Our algorithm is compatible with these approximate algorithms. In the following, we revise the regret bounds for the case when the approximate algorithm in \citep{karmarkar1982efficient} is used to solve the offline bin packing problem in Algorithm \ref{alg:overflow}. Their proofs are identical to the previous case, just with some additional logarithmic terms to account for the approximation gap. We note that the same compatibility to approximate algorithms holds for the algorithm proposed by \cite{rhee1993line2, rhee1993line1}, but the algorithm by \cite{banerjee2020uniform} does not permit an approximate oracle to the offline bin packing problem.

\begin{theorem}[Upper bound for stochastic model with approximate oracle]
Suppose $\mathcal{X}_T=\{X_1,\dots,X_T\}$ is a set of i.i.d. samples from distribution $F$. The following inequality holds for all $T\ge 1,$
$$ \sup_{\mathrm{supp}(F)\subset [0,1]} \{\E [N^{\tilde{\pi}}(\mathcal{X}_T) ]- \E[\mathrm{OPT}(\mathcal{X}_T)]\} \leq 10\sqrt{T} + 3\lceil\log_2 T\rceil^3 + 21\lceil\log_2 T\rceil^2 + 54\lceil\log_2 T\rceil + 13$$
where the expectation is taken with respect to the distribution $F$ and $\tilde{\pi}$ denotes the online packing scheme specified by Algorithm \ref{alg:overflow} but with an approximate offline oracle given by \citep{karmarkar1982efficient}.
\end{theorem}

\begin{theorem}[Upper bound for random permutation model with approximate oracle]
Suppose $\mathcal{X}_T = \{X_1,\dots,X_T\}$ is a randomly permuted sequence of the ground set $\mathcal{X}'_T=\{X_1',\dots,X_T'\}$. The following inequality holds for all $T\ge 1,$
$$ \E[N^{\tilde{\pi}}(\mathcal{X}_T)]\leq \mathrm{OPT}(\mathcal{X}_T') + 12\sqrt{T} + 3\lceil\log_2 T\rceil^3 + 17\lceil\log_2 T\rceil^2 + 42\lceil\log_2 T\rceil + 1$$
where the expectation is taken with respect to the random permutation and $\tilde{\pi}$ denotes the online packing scheme specified by Algorithm \ref{alg:overflow} but with an approximate offline oracle given by \citep{karmarkar1982efficient}.
\end{theorem}

\section{Numerical Experiments and Discussions}

\label{sec_experi}

The results in the previous sections show that the knowledge of $T$ can be utilized to improve  algorithm performance in general. However, when $T$ is unknown but Algorithm \ref{alg:overflow} is implemented blindly, or equivalently, when the algorithm is suddenly terminated at certain timestamp, it may suffer a linear regret. This is a common problem for algorithms that utilize the knowledge of $T$ \citep{shor1986average, asgeirsson2002closed}. Nevertheless, we remark that the implication of Algorithm \ref{alg:overflow} goes beyond the improvements on regret and benchmark analyses (done in the previous sections) from two aspects. First, it tells that the lower bound of $\Omega(\sqrt{T \log T})$ can be overcome under certain setting/problem instance.  Second, it shows the effectiveness of adaptive algorithm that utilizes the past observation to inform future decision making \citep{gupta2020random}. 

For the first aspect, the intuition is aligned with the previous work on integer-valued size distributions \citep{csirik2006sum, gupta2020interior} where $O(B\sqrt{T})$ regret is achieved. To give one more example, when the item sizes follow a uniform distribution on $\{1/K, \dots, K/K\}$, \cite{coffman1991fundamental, coffman1997bin} give a lower bound of $\Omega(\sqrt{T \log K})$ and show that the Best Fit algorithm achieves an $O(\sqrt{T} \log K)$ upper bound. For the second aspect, we utilize the idea of adaptive design and present an LP-based adaptive algorithm for the case of integer-valued size distribution. The algorithm demonstrates good numerical performance, even giving bounded regret under some distributions, and it exhibits good potential to be extended to real-valued distributions. While the previous sections show that the knowledge $T$ is useful in improving algorithm performance generally, the numerical experiments show that this is also true if we restrict our attention to integer-valued size distributions. We leave its analysis as a future question.

Suppose the size of the bins is $B\in \mathbb{N}$ and the item sizes are all integer-valued. Specifically, the distribution $F$ of item size is described by
$\prob(X_t=s_j)=p_j$ for $j=1,...,J$ and $t=1,...,T$ where $s_j \in \{1,...,B-1\}.$ To describe the new algorithm, we first introduce the a level-based LP formulation for the integer-valued bin packing problem (See \citep{csirik2006sum, gupta2020interior}). 
\begin{align*}
    \text{OPT}_I(F) \coloneqq \min_{\bm{v}} \  & \sum_{j=1}^{B-1} v(j,0),\\
    \text{s.t. } \  & \sum_{j=1}^h v(j,h-j) \ge \sum_{j=1}^{B-h} v(j,h), \text{ \ for any \ } h\in \{1,...,B-1\},\\
    & \sum_{h=0}^{B-j} v(j,h) = p_j, \text{ \ for any \ } j \in \{1,...,J\},\\
    & v(j,h)=0, \text{ \ for any \ } h \ge B-j+1, \\
    & v(j,h)\ge 0, \text{ \ for any \ } j,h.
\end{align*}
Here the decision variable $v(j,h)$ describes the fraction of type-$j$ items placed on the level $h$ of a bin. The first inequality constraint describes ``no floating items'', i.e., the total fraction of items that sit on the level $h$ ($\sum_{j=1}^{B-h} v(j,h)$) should not be more than the total fraction of items that end at level $h$ ($\sum_{j=1}^h v(j,h-j)$) since any item of former type must be placed just above an item of the latter type in its bin. The second equality constraint describes a long-term balance that the total fraction of type-$j$ items should equal to $p_j.$ The third and fourth constraints are naturally implied by the rule of the bin packing problem. 

For the online bin packing problem, the distribution $F$ is unknown. We denote the number of bins on the level $h$ at time $t$ as $N_t(h)$. In other words, at time $t$, there are $N_t(h)$ bins that hold items of a total size of $h$. To estimate the parameters of the distribution $F$, we construct estimators
$$\hat{p}_{tj}=\frac{\sum_{t'=1}^t I(X_{t'}=s_j)}{t}$$
for each time $t$.

Algorithm \ref{alg:nl} describes an LP-based adaptive algorithm, where at each time $t$, it solves an LP and uses the LP's optimal solution to guide the packing of the current item $X_t.$ The adaptive LP $\text{OPT}_{t}$ in Algorithm \ref{alg:nl} makes a few changes to the static LP problem $\text{OPT}_{I}(F).$ For the first constraint, it takes account of the current bin status $N_{t}(h).$ For the second constraint, it ensures that there will be a slot for the current item $X_t.$ For the third constraint, it plugs in the estimators for the true probability $p_{tj}.$ These changes all reflect the notion of adaptivity: the algorithm is adaptive to both the current packing status and the underlying distribution. Also, the decision variables convey a meaning of counts over the remaining horizon, rather than a long-term fraction. After solving the LP, the algorithm normalizes the probability and packs the current item accordingly.

\begin{algorithm}[ht!]
\caption{LP-based Adaptive Algorithm}
\label{alg:nl}
\begin{algorithmic}[1]
\Require Time horizon $T$. 
\Ensure A packing scheme for the $T$ items.
\For{time step $t=1,\dots, T$}
        \State Observe item size $X_t$.
        \State Solve the following adaptive LP and obtain an optimal solution $v^*_t(j,h)$.
\begin{align*}
  \text{OPT}_{t} \coloneqq \min_{\bm{u}}\  & \sum_{j=1}^{B-1} v(j,0)\\
    \text{s.t.} \  & N_{t}(h)+\sum_{j=1}^h v(j,h-j) \ge \sum_{j=1}^{B-h} v(j,h), \text{ \ for any \ } h \in \{1,...,B-1\}\\
    & \sum_{h \in H_t}v(X_t, h) \geq \frac{1}{T}, \text{ \ for \ } H_t \coloneqq \{0\}\cup \{h:N_t(h) > 0\}\\
    & \sum_{h=0}^{B-j} v(j,h) = (T-t+1)\hat{p}_{tj}, \text{ \ for any \ } j\in \{1,...,J\}\\
    & v(j,h)=0, \text{ \ for any \ } h \ge B-j+1 \\
    & v(j,h)\ge 0, \text{ \ for any \ } j,h
\end{align*}
\State Normalize the vector $(v^*_t(X_t, h))_{h\in H_t}$ into a probability distribution. 
\State Pack the newly arrived item $s_t$ according to this distribution.
\EndFor
\end{algorithmic}
\end{algorithm}

Figure \ref{fig1} reproduces the numerical experiments done in \citep{gupta2020interior} and it compares the performance of Algorithm \ref{alg:nl} against the two existing algorithms of \citep{csirik2006sum} and \citep{gupta2020interior}. Specifically, we test the three algorithms over different horizon length $T=10, 20, 50, 100, 200, 500, 1000, 2000$, and for each combination of horizon and distribution, the regret of an algorithm is reported based on the average over $100$ simulation trials. The three distributions used for the numerical experiments are where $p_i$ represents the probability of an item with size $i$:
 
Bounded waste: $B = 9$, $F = (p_2 = \frac{35}{48}, p_3 = \frac{13}{48})$.

Perfectly packable: $B = 10$, $F = (p_1 = p_3 = p_5 = \frac{1}{4}, p_4 = p_8 = \frac{1}{8})$.

Linear waste: $B = 10$, $F = (p_3 = p_4 = p_5 = p_8 = \frac{1}{4})$.\\
We observe a dominating performance of Algorithm \ref{alg:nl} over all three types of distributions. Note that the previous two algorithms are distribution-oblivious, meaning that they do not perform learning/estimation of the distribution. The dominance of Algorithm \ref{alg:nl} can be accounted by the knowledge of $T$ and the adaptivity -- distribution learning done throughout the procedure. 

\begin{figure}[ht!]
    \centering
    \includegraphics[scale=0.35]{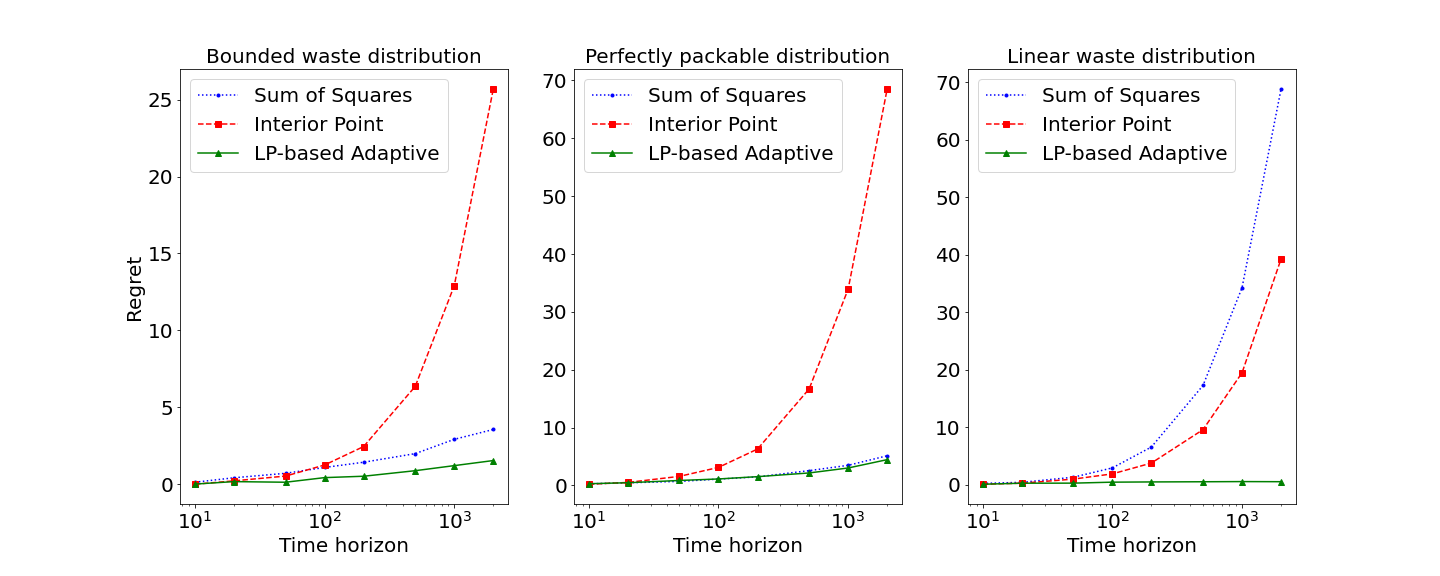}
    \caption{Performance comparison under different distributions}
    \label{fig1}
\end{figure}

Figure \ref{fig2} focuses on Algorithm \ref{alg:nl} and reports its performance if the online procedure is suddenly terminated. The experiment provides some visualization for the temporal dynamic of the algorithm. Specifically, Algorithm  \ref{alg:nl} is implemented with horizon $T=1000$ but terminated at some time $t\le1000$, and then the regret up to the termination time $t$ is reported. Different distributions with different parameters are tested. It is intuitive that when the termination time is close to $T$, the regret is small; on the opposite, if the algorithm is terminated halfway, the regret becomes larger. If we initially inform the algorithm that horizon $T=1000$, it may take aggressive decisions in opening many bins during the halfway (such as $t=500$) for the sake of some long-term good packing scheme, and this will result in a larger regret if the algorithm is suddenly terminated. However, the regret of Algorithm \ref{alg:nl} is relatively small, even when suddenly terminated; this can be seen from a comparison between the left panel of Figure \ref{fig2} and the performance of the other two algorithms in Figure \ref{fig1}.

\begin{figure}[ht!]
    \centering
    \includegraphics[scale=0.35]{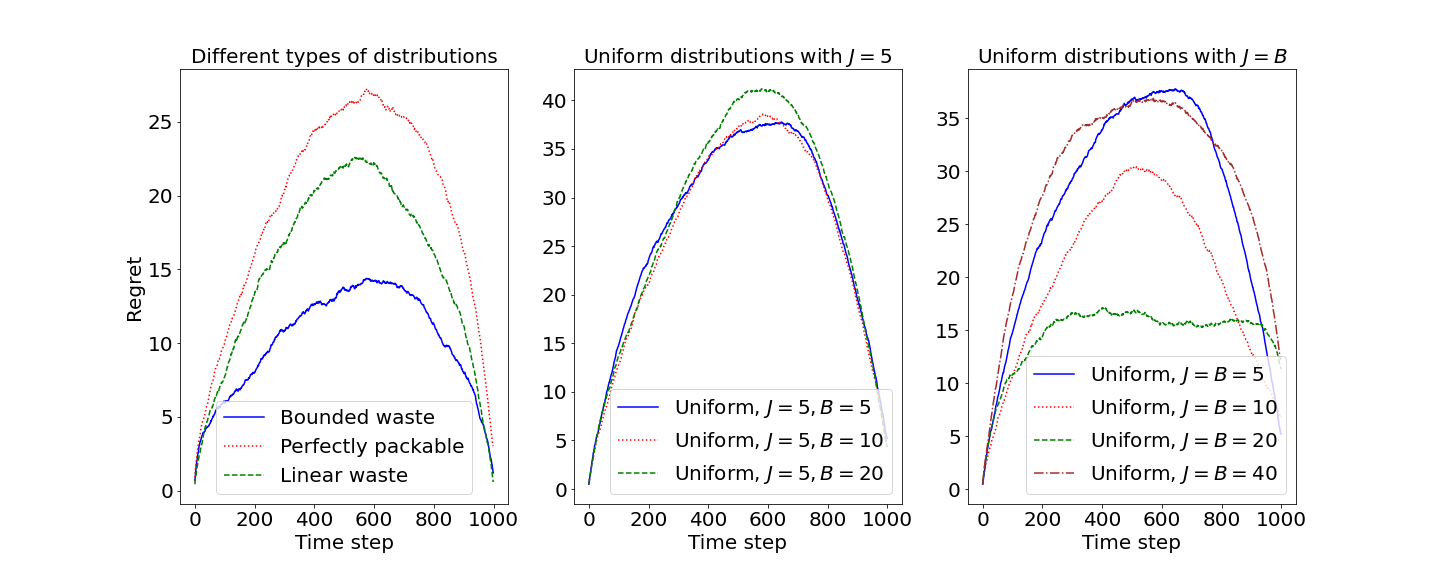}
    \caption{Temporal performance of Algorithm \ref{alg:nl} under different uniform distributions}
    \label{fig2}
\end{figure}

Figure \ref{fig3} explores the dependency of the regret of Algorithm \ref{alg:nl} on $J$ and $B$. We can see that the regret grows mildly with the number of item categories $J$ but seems to bear no dependency on $B$. This indicates that the $O(B\sqrt{T})$ bound in \citep{csirik2006sum, gupta2020interior} may be improvable to $O(\sqrt{JT})$. Moreover, the mild growth rate with respect to $J$ may point to a possibility to adapt Algorithm \ref{alg:nl} for a real-valued item size distribution through some discretization.

\begin{figure}[ht!]
    \centering
    \includegraphics[scale=0.35]{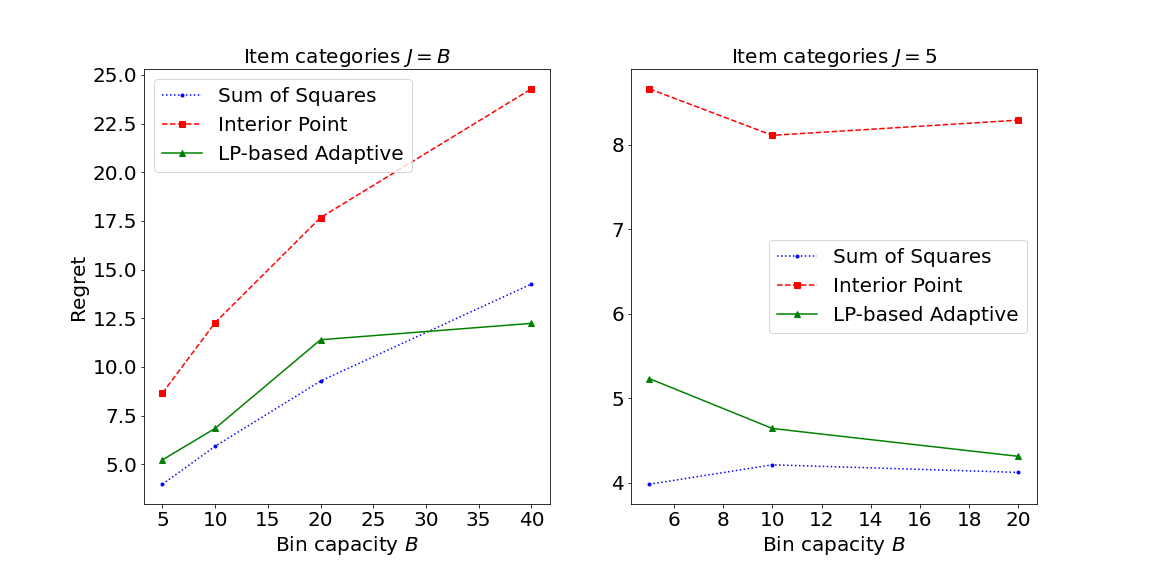}
    \caption{Performance of Algorithm \ref{alg:nl} with respect to different bin capacity $B$ and number of item category $J$. Horizon $T=1000.$}
    \label{fig3}
\end{figure}

\bibliographystyle{informs2014} 
\bibliography{main.bib}

\end{document}